\documentclass[english]{article}
\usepackage[T1]{fontenc}
\usepackage[latin9]{inputenc}
\usepackage{amsmath,amssymb,amsthm,,color}
\usepackage{babel}
\newtheorem{thm}{Theorem}{}{}
\newtheorem{lem}{Lemma}{}{}
{}{}
\newtheorem{prot}{Protocol}{}{}

\def\bra#1{\mathinner{\langle{#1}|}}
\def\inp#1#2{\bra{#1}\mathinner{{#2}\rangle}}
\def\inp2#1#2{\langle{#1} , \mathinner{{#2}\rangle}}

\newcommand{\beq}{\begin{equation}}
\newcommand{\eeq}{\end{equation}}
\newcommand{\commentout}[1]{}
\newcommand{\be}{\begin{enumerate}}
\newcommand{\ee}{\end{enumerate}}

\definecolor{blk}{rgb}{0,0,0}
\definecolor{blu}{rgb}{0.2,0.2,0.7}
\definecolor{red}{rgb}{0.7,0.2,0}

\begin{document}

\title{Information and communication in polygon theories.}

\author{Serge Massar\footnotemark[1] and Manas K. Patra\footnotemark[1]}

\renewcommand{\thefootnote}{\fnsymbol{footnote}}
\footnotetext[1]{Laboratoire d'Information Quantique, CP225, Department of Physics, Universit{\' e} libre de Bruxelles (ULB), Av. F. D. Roosevelt 50, B-1050 Bruxelles, Belgium. \texttt{\{smassar,manas.kumar.patra\}@ulb.ac.be}}

\maketitle

\begin{abstract}
Generalized probabilistic theories (GPT) provide a framework in which one can formulate physical theories that includes classical and quantum theories, but also many other alternative theories. In order to compare different GPTs, we advocate an approach in which one views a state in a GPT as a resource, and quantifies the cost of interconverting between different such resources. We illustrate this approach on polygon theories (Janotta et al. New J. Phys 13, 063024, 2011)  that interpolate (as the number $n$ of edges of the polygon increases) between a classical trit (when $n=3$) and a real quantum bit (when $n=\infty$). Our main results are that simulating the transmission of a single $n$-gon state requires more than one qubit, or more than $log(log(n))$ bits, and that $n$-gon states with $n$ odd cannot be simulated by $n'$-gon states with $n'$ even (for all $n,n'$). These results are obtained by showing that the classical capacity of a single $n$-gon state with $n$ even is 1 bit, whereas it is larger than 1 bit when $n$ is odd; by showing that transmitting a single $n$-gon state with $n$ even violates information causality; and by showing studying the communication complexity cost of  the nondeterministic not equal function using $n$-gon states.
\end{abstract}

\section{Introduction.}

The formalism of Generalised Probabilistic Theories (GPT) introduced several decades ago~\cite{Mackey63,Davies70,Edwards71,Foulis81,Ludwig85} provides a framework for studying generalisations of classical and quantum theory. The development of quantum information theory \cite{NielsenChuang10}  motivated a renewed interest in GPT \cite{Hardy01,Barrett07,Chiribella10,Chiribella11,Masanes11}, with the aim of understanding from the point of view of information processing, what makes quantum theory special \cite{Fuchs02,Brassard05}. Some phenomena considered as uniquely quantum, such as no-cloning and no-broadcasting, are generic features of GPT, see e.g.~\cite{Barrett07,Barnum06,Barnum07}.
Considerable attention has also focused on theories more non-local than quantum \cite{PopescuRohrlich94,VanDam05,BLMPPR05,BarrettPironio05,BBLMTU06}, as no--signaling theories also have typically quantum properties such as intrinsic randomness, secret key generation, no-cloning, see e.g.~\cite{Masanes06,BHK,Acin06}. On the other hand certain quantum features, such as continuity of operations \cite{Barrett07}, entanglement swapping and teleportation\cite{Barnum08,Barnum09}, non triviality of communication complexity \cite{VanDam05,BBLMTU06},  uniqueness of entropy \cite{Barnum10,ShortWehner10} and  data processing inequalities\cite{Barnum10} do not hold in many of these theories.

The study of quantum communication and entanglement has benefited from a quantitative approach in 
which one quantifies the cost of interconverting between different resources. 
For instance dense coding \cite{BennettWiesner} and teleportation \cite{BBCJPW} show how quantum communication, entanglement, and classical communication can be interconverted. The amount of classical communication required to simulate the communication of a single qubit and the non--local correlations produced by a singlet have been established \cite{Maudlin92, BCT99,Steiner00,Csirik02,CGM00,MBCC01,BaconToner03A,TonerBacon03B}, with lower bounds
coming from communication complexity, see the review \cite{BCMdW10} for the case of many qubits, and \cite{MBCC01} for the case of a single qubit. These iterconversions show how much more powerful one resource is than another, and under what conditions are two resources equivalent.

We advocate here that a similar approach will be extremely fruitful in the study of GPT. In this approach one views the communication of a GPT state as a resource, and one wishes to quantify how much this resource is worth. 
What is the cost of replacing the communication of one GPT state by classical bits, by qubits, by states of another GPT?
This question is not altogether new, but has been considered essentially only in the context of 
interconversions between non local correlations  \cite{BLMPPR05,BarrettPironio05,JonesMasanes05,BrussnerSkrzypczyk09,FWW09}. 
However this approach can be applied to all GPTs, and to simpler problems such as one way communication of a GPT state. Prior works in this direction include 
 \cite{FMPT13} in which it was shown that one way communication using a GPT based on the completely positive cone requires exponential classical (and conjectured quantum) communication to simulate, and \cite{BKLS14} in which one way communication of ``hypercube bits'' is discussed.

Quantifying the cost of interconversion between GPTs is realised through two complementary techniques. First one  constructs explicit protocols describing how GPT1 can be simulated by GPT2. Second one proves lower bounds on how many resources of GPT1 are required to simulate GPT2. The second kind of result will generally be obtained by exhibiting an information processing task that requires many resources using GPT1, but can be done cheaply using GPT2. Such tasks are very interesting, because they can discriminate between  theories. They could therefore be used as physical/information theoretic arguments for selecting theories from the large space of GPTs.

Here we illustrate this approach in the context of polygon theories\cite{Janotta11}, so named because their state space is given by a regular polygon.  We refer to them as $n$-gon theories, where $n$ denotes the number of vertices of the polygon. This family  of GPTs interpolates between a classical trit (when $n=3$) and a qubit (when $n=\infty$). The non locality of these theories was studied in \cite{Janotta11}, where it was shown that there are profound differences between polygon theories with odd and even number of vertices:
the odd theories approach Tsirelson's bound from below, while even theories approach it from above. These results were however given without proof of tightness, and required a choice of tensor product. 

Here we consider one way communication with polygon theories. This is a simpler context than \cite{Janotta11}, as the state space is smaller, and it does not require the introduction of a tensor product structure. For this reason one can hope to obtain more detailed results.
Our main results are:
\begin{itemize}
\item For all $n$, simulating the communication of an $n$-gon state cannot be done with a single qubit.
\item  If the two parties do not have any shared randomness, we exhibit a protocol to simulate the communication of an $n$-gon state by sending $\log n$ bits. We also prove a lower bound of
 $\log(\log n)$ bits for this classical simulation.
\item The $n$-gon theories with $n$ odd cannot be simulated by the theories with $n$ even.
\end{itemize}
These results are obtained by studying the classical capacity of  $n$-gon states, the non deterministic communication complexity of the Not Equal function, and the communication complexity of the index function (also known as random access coding) with $n$-gon states. 
In the conclusion we discuss the open questions.

\section{Polygon theories}\label{square}

We recall here the polygon theories introduced in \cite{Janotta11}  to which we refer for further details.
The set of normalised states of  $n-$gon theory has $n$ extremal states:
\begin{equation}\label{eq:state}
\omega_i = \begin{pmatrix}r_n\cos{\frac{2\pi i}{n}} \\ r_n\sin{\frac{2\pi i}{n}}\\ 1 \end{pmatrix}\in \mathbb{R}^3\quad ,\quad r_n=\sqrt{\sec{(\pi/n)}}\quad ,\quad 0\leq i < n
\end{equation}
The space of normalised states is the convex hull of $\{\omega_k\}$. It is a regular $n$-gon, thereby giving their name to these theories.
The space of unnormalised states is the cone $C_n$ generated by the $\{\omega_k\}$.

The space of effects is the dual $C^*_n$ of the cone $C_n$.
For even $n$ the dual cone $C^*_n$ is generated by the extremal effects
\begin{equation}\label{eq:effect0}
e_j= \frac{1}{2}\begin{pmatrix}r_n\cos{\frac{(2j-1)\pi}{n}} \\  r_n\sin{\frac{(2j-1)\pi}{n}} \\ 1 \end{pmatrix}\ .
\end{equation}
In the even case the cone $C_n$ is weakly self dual: the cone $C^*_n$ equals the cone $C_n$ rotated by the angle $\pi/n$.

For odd $n$ the cone $C^*_n=C_n$ is generated by the extremal effects
\begin{equation}\label{eq:effect1}
e_j= \frac{1}{1+r_n^2}\begin{pmatrix} r_n\cos(2\pi j/n)\\r_n\sin(2 \pi j/n)\\1\end{pmatrix}
\end{equation}
For odd $n$ the cone $C_n$ is strongly self dual: $C^*_n=C_n$.

The unit effect is 
$$u=\left( 0, 0, 1\right)\ .$$ 
The normalised states have unit scalar product with the unit: $\langle \omega_k,u\rangle =1$.

A measurement 
$$M=\{f_k\in C^*_n : \sum_k f_k=u \}$$
is a set of effects (of elements of the dual cone) that sum to the unit effect.
The probability of outcome $k$ given normalised state $\omega$ and measurement $M$ is 
$$P(k\vert \omega)=\langle f_k,\omega\rangle\ .$$
In general a measurement can have an arbitrarily large number of effects. However in \cite{FMPT13} it was shown that by refining a measurement, and decomposing a measurement into a convex combination of other measurements, one can restrict to measurements with at most 3 effects all proportional to the extremal effects. (Here 3 is the dimension of the space in which the states and effects are defined). That is we can restrict to measurements of the form:
\begin{equation}
M=\{\lambda_k e_k, k=1,2,3:  0\leq \lambda_k, \sum_{k=1}^3 \lambda_k e_k=u \}\ .
\label{canonicalmeast}
\end{equation}
In appendix \ref{app1} we present in more detail the structure of measurements with 3 extremal effects.

For $n$ even the complement $u -e_i$ of any extremal effect $e_i$ is an extremal effect. So in this case the measurements $M_i = \{e_i, u-e_i=e_{i+n/2}\}$ constitute 2-outcome measurements. (For odd $n$ there do not exist 2 outcome measurements both of whose effects are extremal).

Note that the extremal effects eqs. (\ref{eq:effect0}, \ref{eq:effect1}) are normalised such that
$0\leq \langle e_j,\omega_i\rangle\leq 1$. These inequalities are saturated in the following cases:
\begin{eqnarray}
\langle e_j,\omega_i\rangle=0 &\mbox{when}& i=j+n/2\ ,\ i=j+n/2-1 \quad \mbox{($n$ even)}\\
 & & i=j+(n+1)/2  \ ,\ i=j+(n-1)/2  \quad \mbox{($n$ odd)}\\
\langle e_j,\omega_i\rangle=1 &\mbox{when}& i=j\ ,\ i=j-1 \quad\mbox{($n$ even)}\\
 & & i=j\quad \mbox{($n$ odd)}
\label{satprob}
\end{eqnarray}

There are important differences between polygon theories with odd and even number of states. Some of these differences are discussed in \cite{Janotta11} in the context of non--locality. They will also appear clearly in what follows.

\section{Classical information capacity of polygon theories.}

\subsection{Statement of the result}

We study the classical capacity of the $n$--gon theories. Alice receives  symbols $x$ drawn from a probability distribution $p_x$. Upon receiving symbol $x$, she sends state $\omega(x)$ to Bob, where $\omega(x)$ is a state of $n$--gon theory. Bob carries out a measurement, obtaining outcome $y$. The capacity of the channel is the maximum, over all probability distributions $p_x$, all encodings $\omega(x)$, all measurements, of the mutual information $I(X;Y)$.
More precisely $I(X;Y)$ is, in the asymptotic limit, the number of bits Alice can send to Bob using many copies of the channel, using block coding, but no entanglement between systems.

This general setting can be simplified by noting that Bob's measurement can be taken from the set of canonical measurements eq. (\ref{canonicalmeast}) (this follows from the data processing inequality, and the fact that the mutual information is convex in the probabilities $p_{y\vert x}$) . These measurements have at most 3 outcomes, and hence the capacity of $n$-gon theories is at most $\log 3$ bits. Here we prove stronger results:

\begin{thm}\label{thm1}
When $n$ is even, the classical capacity of $n$--gon theories is exactly 1bit.
\end{thm}

\begin{thm}\label{thm2}
When $n$ is odd, the classical capacity of $n$--gon theories is strictly larger than 1 bit, is equal to $\log 3$ bits when $n=3$, and tends towards $1$ bit when $n\to \infty$.
\end{thm}

To prove these results, we first exhibit communication protocols that satisfy the capacities stated in the theorems.

\begin{prot}
($n$ even).  Alice has two inputs, $x=0,1$, that occur with equal probability $p_0=p_1=1/2$. On input $x=0$ she prepares state $\omega_0$, on input $x=1$ she prepares state $\omega_{n/2}$. Bob carries out the measurement $M=\{e_0, e_{n/2}\}$. 
\end{prot}
One easily checks using eqs. (\ref{satprob}) that in this case the capacity is $1$ bit.

\begin{prot} ($n$ odd). Alice has three inputs, $x=0,1,2$, that occur with  probability $p_0=1/2$, $p_1=p_2=1/4$. On input $x=0$ she prepares state $\omega_0$, on input $x=1$ she prepares state $\omega_{(n-1)/2}$, on input $x=2$ she prepares state $\omega_{(n+1)/2}$. Bob carries out the measurement $M=\{\frac{1}{r_n^2}e_0, \frac{1}{2}e_{(n-1)/2},\frac{1}{2}e_{(n+1)/2}\}$. 
\end{prot}
Using eqs. (\ref{satprob}) one checks that in this case the capacity is strictly larger than $1$ bit. When $n=3$ the capacity of this protocol is $\log 3$, and it decreases monotonically towards $1$ bit as $n$ tends to $\infty$.

We now turn to proving the converses.

\subsection{Upper bounds on the classical capacities.}

As noted above, without loss of generality we can take Bob's measurement $M=\{\lambda_y e_y, y=1,2,3\}$ to have 3 extremal effects, see eq. (\ref{canonicalmeast}). The normalisation condition $\sum_y \lambda_y e_y=u$ implies that $\sum_y \lambda_y=2$ ($n$ even) and  $\sum_y \lambda_y=1+r_n^2$ ($n$ odd). For future notation, we denote $\sum_y \lambda_y=c_n$, where $2\leq c_n\leq3$.
The probability of Bob obtaining outcome $y$ given state $\omega(x)$ is $P(y\vert x)= \lambda_y  \langle e_y ,\omega(x)\rangle \leq \lambda_y$.

The mutual information $I(Y;X)$  is concave in the conditional probabilities $P(y\vert x)$ (i.e. taking a convex combination of two channels can only decrease the capacity). Hence the capacity is maximal at the extreme points of the set of conditional probabilities $P(y\vert x)$. If we enlarge the space of possible channels, we can only increase the capacity.
In particular if we impose only the following conditions:
\begin{eqnarray} 
 P(y\vert x) &\geq& 0\label{E0}\\
P(y\vert x)&\leq& \lambda_y\label{EL}\\
\sum_y P(y\vert x)&=&1\label{EN}\\
\sum_y \lambda_y&=&c_n \label{ENL}
\end{eqnarray}
with arbitrary alphabet size and number of measurement outcomes equal to three $y=1,2,3$, then we include all channels realised by polygon theories (as well as some other channels).
Hence the capacity of the channels defined by  eqs. (\ref{E0}, \ref{EL}, \ref{EN}, \ref{ENL})  is at least as large as the capacity of  the polygon theories.

Concavity of $I(X;Y)$ implies that the capacity of the channels defined by 
eqs. (\ref{E0}, \ref{EL}, \ref{EN}, \ref{ENL})  describe (for fixed alphabet size $\vert X \vert$) a polytope whose vertices are the extreme points. Concavity of $I(X;Y)$ implies that the capacity of the corresponding channels   will  be maximum at a vertex of this polytope. It is therefore sufficient to determine the capacities of the channels defined by the vertices of this polytope.

\begin{lem}\label{lemvertex}
The vertices of the polytope defined by eqs.  (\ref{E0},\ref{EN},\ref{EL},\ref{ENL}) have the following properties: Either
at least one of the $\lambda_y=0$; or 
when $2<c_n\leq 3$, the vertices have (up to a permutation of the $y$'s) the form 
$\lambda_1=c_n-2$, $\lambda_2=\lambda_3=1$, and all inputs $x$  give rise to one of the following four output distributions:
\begin{eqnarray}\label{probdist1}
&P(1\vert x_1)=0 \quad , \quad P(2\vert x_1)=0 \quad , \quad P(3\vert x_1)=1&\nonumber\\
&P(1\vert x_2)=0\quad , \quad P(2\vert x_2)=1\quad , \quad P(3\vert x_2)=0&\nonumber\\
&P(1\vert x_3)=c_n-2\quad , \quad P(2\vert x_3)=0\quad , \quad P(3\vert x_3)=3-c_n\nonumber\\
&P(1\vert x_4)=c_n-2\quad , \quad P(2\vert x_4)=3-c_n\quad , \quad P(3\vert x_4)=0\ .
\end{eqnarray}
\end{lem}

Note that Lemma \ref{lemvertex} immediately implies the upper bounds stated in Theorems \ref{thm1} and \ref{thm2}, since for $c_n=2$ it implies that the capacity of the channel is at most $1$ bit (since at least one outcome never occurs), and that for $3\geq c_n>2$, the vertices either have capacity less or equal to $1$ bit, or tend towards a channel with capacity $1$ bit as $c_n\to2$.

The proof of lemma \ref{lemvertex} is given in appendix \ref{app2}. An alternative proof that the capacity of the $n$-gon theories with $n$ even is bounded by 1 bit is given in appendix \ref{appdirect}.

\subsection{Optimal coding for Alice}

We also note that it is possible to simplify Alice's coding.

\begin{lem}\label{lemthree}
The maximal information capacity of polygon theories is obtained when
Alice's alphabet has size 3, and each of the inputs is extremal $\omega(x)=\omega_{i(x)}$, $x=1,2,3$.
\end{lem}

The proof of lemma \ref{lemthree} is given in appendix \ref{app3}.

\section{Random access coding and information causality}

We consider the task in which Alice receives uniformly random inputs of $m$ bits, $x_1x_2...x_m$ and Bob receives as input a random index $j\in\{1,2,...,m\}$. Bob's aim is to output a bit $y$  that coincides  with $x_j$. 
Depending on the context, this task goes under the name
random access coding (RAC) \cite{ANTV02}, communication complexity of the index function \cite{BKLS14}, or information causality (IC)\cite{Pawlowski09} .

Here, following \cite{Pawlowski09}, we shall measure the success of the protocol  by
 the average information
\begin{equation} \label{eq:avgInf}
\bar{I} = \sum_jp_jI_j = \sum_j p_jI(X:Y\; |j)
\end{equation}
where $I_j=I(X:Y\; |j)$ is the conditional mutual information between Alice and Bob given that his input is $j$ and $\{p_j\}$ is the probability distribution over Bob's input. 
Note that Fano's inequality \cite{CoverThomas} implies a relation between $I_j$ and the probability $P_j^{succ}$ of Bob successfully decoding Alice's $j$th input through $H(P_j) \geq 1- I_j$.

As shown in \cite{Pawlowski09}, if Alice sends Bob $c$ classical bits, or $c$ qubits, then $I\leq c$ (even if Alice and Bob have shared randomness or shared entanglement). 
It can also be shown \cite{IC2} by adapting slightly  the proof in \cite{Pawlowski09} that if Alice and Bob do not have prior shared entanglement, and Alice sends Bob $q$ quantum bits, then $I\leq q$. The essential property  used in these proofs is that the data processing inequality is valid both for classical \cite{CoverThomas} and quantum information theories \cite{NielsenChuang10}.

The idea that $I$ should be less than the classical capacity of the channel between Alice and Bob is known as information causality. Information causality does not hold in all GPTs, for instance if Alice and Bob share correlations more non local than quantum, and Alice sends Bob classical information. 

Here we consider  information causality in the context of $n$--gon theories (with $n$ even).  Specifically we suppose that Alice and Bob have as prior resource shared randomness. We consider the case $m=2$: Alice receives two bits as input and sends Bob a  state $\omega(x_1,x_2)$. Bob receives as input an index $j\in\{1,2\}$. The aim is for Bob's output to maximize the quantity eq. (\ref{eq:avgInf}). 

Since the classical capacity of $n$--gon theories with $n$ even is 1 bit (theorem \ref{thm1}) one would expect that $I=1$. We show here that this intuition is wrong.

\begin{thm}\label{thm3}
When $n$ is even, sending a single state of $n$-gon theory, achieves  $I>1$.
\end{thm}

We conjecture that the maximum achievable value of $I$ decreases monotonically  from $I=2$ when $n=4$  to $I=1$ when $n=\infty$ (because the limiting cases ($n=4$ and $n=\infty$) are known and the explicit protocol exhibits this monotonicity).

\begin{prot}
($n$ even).  Alice receives two bits $x_0$ and $x_1$ as input. She prepares the state
$\omega_{x_0 n/2+x_1+x_0x_1}$. Bob measures in the basis $M_1 =\{e_1,e_{n/2+1}\}$ and $M_0 =\{e_0,e_{n/2}\}$ on his inputs 0 and 1 respectively.
\end{prot}

Using this protocol, when $j=0$, Bob learns the value of $x_0$ with certainty ($P(y=x_0\vert j=0)=1$); while when $j=1$ 
Bob's success probability  is easily computed to be 
$P_n = 1 - \cos{(2\pi/n)}/2 > 1/2$. Therefore, $I > 1$.

\section{Nondeterministic Not Equal Function}

We recall a result from classical communication complexity (see \cite{KN98} for a review of the field).
Consider the NOT-EQUAL function, 
$F_{NE}: \{0,1\}^k\times\{0,1\}^k\to \{0,1\}$ defined as
\begin{eqnarray}
F_{NE}(x,y)&=& 1\mbox{ if } x\neq y\nonumber\\
& = & 0\mbox{ if } x=y
\end{eqnarray}
Suppose that Alice and Bob are given $x$ and $y$, respectively,
as inputs and their goal is to evaluate $F_{NE}(x,y)$ in
the following weak sense. Bob should output a bit $b$ that is
distributed so that if $F_{NE}(x,y)=0$, then $Pr[b=1]=0$; whereas if
$F_{NE}(x,y)=1$, then $Pr[b=1]>0$. In addition, assume that Alice
and Bob have no prior shared randomness. Such a
protocol can be regarded as a nondeterministic
protocol for the $F_{NE}$ function. 
The following lower bound on the amount of classical communication required by Alice and Bob to achieve this is
 well known (see \cite{KN98}):

\begin{lem} Any nondeterministic classical protocol for
computing the function $F_{NE}$ requires at least $log_2(k)$ bits of
communication.
\end{lem}

On the other hand we recall the result \cite{MBCC01}:
\begin{lem} There exists a nondeterministic quantum protocol for
computing the function $F_{NE}$ that uses one qubit of
communication.
\end{lem}

We now show how polygon theories can be used to compute nondeterministically the NE function.

\begin{lem} By sending a single state of $n$-gon theory, one can achieve a nondeterministic protocol for computing the function $F_{NE}: \{0,1,...,n-1\}\times\{0,1,...,n-1\}\to \{0,1\}$.
\end{lem}

This is achieved by the following protocol.

\begin{prot}
Upon receiving input $x\in \{0,1,...,n-1\}$, Alice sends Bob state $\omega_x$.
Upon receiving input $y\in \{0,1,...,n-1\}$, Bob carries out measurement
$M_y=\{e_y,e_{y+n/2}\}$ ($n$ even) or 
$M_y=\{\{\frac{1}{r_n^2}e_y, \frac{1}{2}e_{y+(n-1)/2},\frac{1}{2}e_{y+(n+1)/2}\}$ ($n$ odd).
Bob outputs $0$ if he obtains outcome $e_y$, and outputs $1$ if he obtains one of the other outcomes.
\end{prot}

One easily checks, using eq. (\ref{satprob}) that this is a nondeterministic protocol for the NE function.

This implies the following lower bound on the classical cost of simulating polygon theories:
\begin{thm} If Alice and Bob do not have a priori shared randomness, then simulating the communication of a single state of $n$-gon theory requires $\log_2 (\log_2 n)$ bits of classical communication.
\end{thm}

On the other hand we have the following result:
\begin{thm} If Alice and Bob do not have prior shared randomness, they can simulate the communication of a single state of $n$-gon theory using $\log_2 n$ bits of classical communication.
\end{thm}

which is achieved using the following simple protocol:

\begin{prot}
Suppose we wish to simulate Alice  sending to Bob state $\omega$, and Bob  carrying out measurement $M=\{f_k\}$ on the state, such that result $k$ is obtained with probability $P(k\vert\omega)=\langle f_k,\omega\rangle$. This can be simulated classically as follows. Alice decomposes $\omega=\sum_{i=0}^{n-1} p_i \omega_i$ into extremal states. She chooses (using local randomness) index $i\in\{0,...,n-1\}$ from the probability distribution $p_i$. She sends index $i$ to Bob. Bob outputs $k$ with probability $P(k\vert i)= \langle f_k,e_i\rangle$.
\end{prot}

\section{Conclusion}

In the present work we view generalised probabilistic theories as resources, and inquire how much of one resource is needed to simulate the other. 
This approach to studying GPT is particularly interesting because it immediately raises a large number of precise, quantitative, questions.
To illustrate this approach we considered the specific case of polygon theories, and more particularly the situation where one $n$--gon state is sent from Alice to Bob.
We have obtained a number of results for this family of theories, but many questions remain open.

Some of the main open questions in the context of $n$-gon theories are:
\begin{enumerate}
\item The transmission of one $n$-gon state cannot be simulated by the transmission of one qubit. How many qubits are in fact needed?
\item We have shown that in the absence of shared randomness, the cost of classically simulating the transmission of an $n$--gon state requires at least $\log_2(\log_2 n)$ bits. We have given a protocol that uses $\log n$ bits, and does not use any local randomness. What is the classical cost of simulating the transmission of an $n$-gon state with no randomness, with local randomness only, with shared randomness? Concerning the last question, it is not obvious whether one can transpose the protocols for the simulation of a single qubit  \cite{Maudlin92, BCT99,Steiner00,Csirik02,CGM00,MBCC01,BaconToner03A,TonerBacon03B}, since  $n$-gon states cannot be simulated by a single qubit.
\item What is the cost of simulating an $n$-gon state by $m$-gon states? Our results on the classical capacities imply that a single $n$-gon states with $n$ odd cannot be simulated by a single $m$-gon state with $m$ even; and our results on the NOT EQUAL function imply that a single $n$-gon state  cannot be simulated by a single $m$-gon state when $m< \log n$. However much tighter bounds may be possible.
\end{enumerate}

\appendix

\section{Extremal 3 effect measurements.}\label{app1}

The following lemma provides the detailed structure of 3--outcome measurements all of whose effects are extremal.
\begin{lem}\label{lem:struct1}
In polygon theories, the measurements with three extremal effects are given by
\begin{eqnarray}\label{eq:struct-extremal}
%\begin{split}
M &=& a\left(\lambda_1 e_{j_1}, \lambda_2 e_{j_2}, \lambda_3 e_{j_3} \right) \text{ where } 0\leq \lambda_i\leq 1 \text{ and }\nonumber\\
\lambda_1 &=& 1 - \cot{\frac{(j_1 -j_2)\pi}{n}}\cot{\frac{(j_3 -j_1)\pi}{n}}\ ,\nonumber\\
\lambda_2 &=& 1 - \cot{\frac{(j_1 -j_2)\pi}{n}}\cot{\frac{(j_2 -j_3)\pi}{n}}\ ,\nonumber\\
\lambda_3 &=& 1 - \cot{\frac{(j_2 -j_3)\pi}{n}}\cot{\frac{(j_3 -j_1)\pi}{n}}\ ,\nonumber\\ 
a &=& 
\begin{cases} 1\; (n \text{ even})\\ (1+r_n^2)/2\; (n \text{ odd})   \end{cases}
%\end{split}
\end{eqnarray}
Further, for odd $n$ there is constant $\delta_n > 0$ such that $\lambda_i \geq \delta_n$ and $\underset{n\rightarrow\infty}{\lim}\delta_n = 0$. 
\end{lem}
\begin{proof}
The coefficients $\lambda_i$ and $a$ follow after a tedious but straightforward calculation. The first part of the last statement follows from the fact that for odd $n$ there do not exist measurements with 2 extremal effects. The limiting value is easily computed using the formulas in \eqref{eq:struct-extremal} and taking $j_1= 0$, $j_3=(n+1)/2$ and $j_2= (n\pm 1)/4$, whichever is an integer. 
\end{proof}

\section{Proof of lemma \ref{lemvertex}.}\label{app2}

Here we prove lemma \ref{lemvertex} . 

\begin{proof}
The number of variables in eqs. (\ref{E0},\ref{EN},\ref{EL},\ref{ENL}) is $3\vert X\vert+3$. Taking into account the equalities  eqs. (\ref{EN},\ref{ENL}), there are
$2\vert X\vert+2$ independent variables. A point $\{P(x\vert y), \lambda_y : x=1,...,\vert X\vert, y=1,2,3\}$ is a vertex if all inequalities and equalities in eqs. (\ref{E0},\ref{EN},\ref{EL},\ref{ENL}) are satisfied, and $2\vert X\vert+2$ linearly independent inequalities eqs. (\ref{E0},\ref{EN}) are saturated such that the values of all variables $P(x\vert y)$, $\lambda_y$ are fixed.

Let us consider a specific input  $x=x_0$. For this value of $x_0$, there are 6 inequalities eqs. (\ref{E0},\ref{EN}) that could be saturated. A sufficient number of them must be saturated that all values of $P(y \vert x_0)$ are fixed. We now enumerate all combinations of inequalities that can be saturated, such that they fix all the values of the $P(y\vert x_0)$. We give the inequalities that are saturated, and the implications for the other variables. For ease of notation, we give the inequalities that are saturated up to a permutation of the $y$'s.
\begin{enumerate}
\setcounter{enumi}{1}
\item Exactly 2 inequalities are saturated
\begin{enumerate}
\item $P(1\vert x_0)=P(2\vert x_0)=0$ implies $P(3\vert x_0)=1$ and $\lambda_3>1$.\\
\item $P(1\vert x_0)=0$, $P(2\vert x_0)=\lambda_2$ implies $P(3\vert x_0)=1-\lambda_2$, $0<\lambda_2<1$, $\lambda_3>1-\lambda_2$.\\
\item $P(1\vert x_0)=\lambda_1$, $P(2\vert x_0)=\lambda_2$ implies $P(3\vert x_0)=1-\lambda_1-\lambda_2$, $0<\lambda_1<1$, $0<\lambda_2<1$, $\lambda_3>1-\lambda_1-\lambda_2$.\\
\end{enumerate}
\item Exactly 3 inequalities are saturated
\begin{enumerate}
\item $P(1\vert x_0)=P(2\vert x_0)=0$, $\lambda_3=P(3\vert x_0)=1$ implies  $\lambda_1>0$, $\lambda_2>0$, $\lambda_1+\lambda_2=c_n-1$.\\
\item $\lambda_1=P(1\vert x_0)=P(2\vert x_0)=0$ implies $P(3\vert x_0)=1$, $\lambda_3>1$.\\
\item $\lambda_1=P(1\vert x_0)=0$, $P(2\vert x_0)=\lambda_2$ implies $P(3\vert x_0)=1-\lambda_2$,  $1>\lambda_2>0$, $\lambda_3>1-\lambda_2$.\\
\item $P(1\vert x_0)=0$, $P(2\vert x_0)=\lambda_2$, $P(3\vert x_0)=\lambda_3$ implies $\lambda_1>0$, $\lambda_2+\lambda_3=1$.
\end{enumerate}
\item Exactly 4 inequalities are saturated
\begin{enumerate}
\item $\lambda_1=\lambda_2=P(1\vert x_0)=P(2\vert x_0)=0$ implies $P(3\vert x_0)=1$, $\lambda_3=c_n$.\\
\item
$\lambda_1=P(1\vert x_0)=P(2\vert x_0)=0$, $P(3\vert x_0)=\lambda_3$ implies $P(3\vert x_0)=1$, $\lambda_3=1$, $\lambda_2=c_n-1$.
\end{enumerate}
\end{enumerate}

Now we turn to the variables $\lambda_y$. At a vertex it must be that the saturated inequalities fix the values of all the $\lambda_y$. We wish to find the vertices for which none of the $\lambda_y$'s vanish.
Note that cases 3b,c and 4a,b fix at least one of the $\lambda_y$'s to zero, hence we cannot use these cases.
Cases 2a,b,c do not fix any of the $\lambda_y$'s.
Case 3a fixes $\lambda_3=1$. Case 3d fixes that $\lambda_2+\lambda_3=1$. 

Case 3d cannot  by itself fix the values of the $\lambda_y$'s. Hence we must use at least once (say for variable $x_1$) case 3a, and we have $\lambda_3=1$, $\lambda_1>0$, $\lambda_2>0$, $\lambda_1+\lambda_2=c_n-1$.

We now wish to to fix $\lambda_1$ and $\lambda_2$ (both different from zero). To this end we must use either cases 3a or 3d (for another value of $x$, say $x_2$).

When $c_n=2$, one easily checks that using cases 3a,d to fix another value of $\lambda$ implies $\lambda_3=1$, $\lambda_2=1$, $\lambda_1=0$ (or a permutation thereof), and therefore one of the $\lambda_y$'s is equal to zero.

When $c_n>2$, the only way to fix all the $\lambda_y$'s without one of them vanishing is to use case 3a again to fix $\lambda_2=1$. Hence we have $\lambda_3=1$, $\lambda_2=1$, $\lambda_1=c_n-2$. 

Finally there may be additional values of $x$ (say $x_3$ and $x_4$) for which we can use case 2b (none of the other cases are compatible with these values  of the $\lambda_y$'s).

We thus obtain the four probability distributions given in eq. (\ref{probdist1}).
\end{proof}

\section{Direct proof of  the upper bound in Theorem \ref{thm1}}\label{appdirect}

In this section we provide a direct proof that the classical capacity of $n$-gon theories with $n$ even is bounded by 1 bit. Recall that the 
probability matrix $P =P(y|x)$ of $n$-gon theories, $n$ even,
obeys eqs. (\ref{E0},\ref{EL},\ref{EN},\ref{ENL}) and that we can restrict the size of the input and output alphabets to three ($x,y=1,2,3$).

{\bf Claim.} Any matrix $P$ satisfying eqs. (\ref{E0},\ref{EL},\ref{EN},\ref{ENL}) with $c_n=2$ can be written as a convex combination of the following matrices. 
\begin{eqnarray}
P_1 &=& \begin{pmatrix} u_1 & u_2 & u_3 \\  1-u_1 & 1-u_2 & 1-u_3\\ 0 & 0 & 0\end{pmatrix}\ ,\nonumber\\
 P_2 &=&  \begin{pmatrix} v_1 & v_2 & v_3 \\  0 & 0 & 0 \\ 1-v_1 & 1-v_2 & 1-v_3 \end{pmatrix} \ ,\nonumber\\
 P_3 &=&  \begin{pmatrix} 0 & 0 & 0 \\ w_1 & w_2 & w_3 \\  1-w_1 & 1-w_2 & 1-w_3\end{pmatrix}\ ,
\end{eqnarray}
where $0\leq u_i,v_i, w_i \leq 1$, i.e. $P=q_1 P_1 + q_2 P_2 + q_3 P_3$  with $0\leq q_1, q_2, q_3\leq 1$ and  $q_1+q_2+q_3=1$.

That is for $x= 1,2,3$ we can write
\begin{eqnarray}\label{eq:EvenConvComb}
P(1\vert x) &=& q_1u_x+q_2v_x\ , \nonumber\\
P(2\vert x)  &=& q_1(1-u_x)+q_3w_x\ ,\nonumber\\
P(3\vert x)  &=& q_2(1-v_x)+q_3(1-w_x)\ .
\end{eqnarray}

\begin{proof}
Fix the values of $\lambda_1,\lambda_2,\lambda_3$. We shall show that  we can satisfy eqs. (\ref{eq:EvenConvComb}) with the choice
\[ q_1= 1 - \lambda_3,\; q_2= 1-\lambda_2 \text{ and } q_3 = 1-\lambda_1\ .\]

Note that since $P(1\vert x)+P(2\vert x) +P(3\vert x) = q_1+q_2+q_3=1$, there is at least one value of $y$ such that $P(y\vert x) \leq q_y$. Without loss of generality we take $P(1\vert x)\leq q_1$.

Note that we only have to satisfy only two of the equations in \eqref{eq:EvenConvComb} since the third is automatically satisfied by normalization.
The following relations are necessary and sufficient for the first two equations in \eqref{eq:EvenConvComb} to be satisfiable.
\begin{eqnarray}\label{eq:nec-suf1}
0 \leq& P(1\vert x)-q_1u_j & \leq q_2 \ ,\nonumber\\ 
q_1(1-u_j) \leq& P(2\vert x) &\leq q_1(1-u_j)+q_3\ .
\end{eqnarray}
Clearly these conditions are necessary. To prove sufficiency (assuming $P(1\vert x)\leq q_1$) if the relations in the first line are true for some choice of $u_j$ then $0\leq v_j=  (p_{1j}-q_1u_j)/q_2\leq 1$. Similarly if the second line holds then $0 \leq w_j = (p_{2j}- q_1(1-u_j))/q_3\leq 1$.

Next, it is easily verified that the relations (\ref{eq:nec-suf1}) are equivalent to 
\begin{eqnarray}
 \max\{0,\frac{P(1\vert x)}{q_1} -\frac{q_2}{q_1}\}  \leq& u_j& \leq \frac{P(1\vert x)}{q_1} \ ,\label{eq:cond1}\\
 1-\frac{P(2\vert x)}{q_1}  \leq &u_j& \leq \min\{1, 1-\frac{P(2\vert x)}{q_1}+\frac{q_3}{q_1}\}\ . \label{eq:cond2}
\end{eqnarray}
Eqs. (\ref{eq:cond1}) and (\ref{eq:cond2}) are separately satisfiable. In order that they be  simultaneously satisfiable it is necessary and sufficient that 
\begin{eqnarray}
\max\{0,\frac{P(1\vert x)}{q_1} -\frac{q_2}{q_1}\}&\leq& \min\{1,1-\frac{P(2\vert x)}{q_1}+\frac{q_3}{q_1}\}\ , \label{eq:con21}\\
1-\frac{P(2\vert x)}{q_1} &\leq& \frac{P(1\vert x)}{q_1} \ .\label{eq:con22}
\end{eqnarray}

The first condition can be shown to be satisfiable by showing that both of the expressions on the left  is $\leq$ to both on the right. From eqs. (\ref{eq:EvenConvComb}) it follows that $P(2\vert x)\leq q_1+q_3$, and hence
that $0 \leq (q_1+q_3-P(2\vert x))/q_1$. Also $\frac{P(1\vert x)}{q_1} -\frac{q_2}{q_1} \leq 1$ due to our assumption that $P(1\vert x) \leq q_1$. Finally  $\frac{P(1\vert x)}{q_1} -\frac{q_2}{q_1} \leq 1-\frac{P(2\vert x)}{q_1}+\frac{q_3}{q_1}$ since this condition is simply $P(1\vert x)+P(2\vert x) \leq q_1+q_2+q_3=1$ which is always true. The second condition (\ref{eq:con22}) is equivalent $q_1 \leq P(1\vert x)+P(2\vert x)$ which as noted above is also true due to the choice of $q_i$. 
\end{proof}

Since the mutual information $I(X:Y)$ is a concave function of the conditional probabilities $P(y|x)$ for fixed input distribution~\cite{CoverThomas}, and denoting by $I_k$ mutual information corresponding to conditional probability matrix $P_k$, we have 
\[ I(X:Y) \leq q_1I_1+q_2I_2+q_3I_3 \leq 1\ .\]
The last inequality follows from the fact that the matrices $P_i$ are essentially 2-dimensional.

\section{Proof of lemma \ref{lemthree}.}\label{app3}

Here we prove lemma \ref{lemthree} . 

\begin{proof}
For simplicity of notation, we denote a channel by $C=\{\omega_x,p_x,e_y\}$, i.e. the collection of states, probabilities, and effects that are used.

The argument is related to that used in \cite{FMPT13} to show that extremal measurements have at most 3 effects. To prove it, 
first note that if $\omega(x)$ is not extremal, we can decompose it into extremal states $\omega(x)=\sum_i p_{i\vert x} \omega_i$. 
We can then compare the two channels $C=\{\omega_x,p_x,e_y\}$ and $C'=\{\omega_i,p_{xi}=p_{i\vert x}p_x ,e_y\}$. Using the chain rule for mutual information, we have
$I(Y;XI)=I(Y;X)+I(Y;I\vert X)$ \cite{CoverThomas} where $I(Y;X)$ is the capacity of the channel $C$ and $I(Y;XI)$ is the capacity of the channel $C'$.  Hence the capacity of the channel $C'$ is larger than the capacity of the channel $C$.

We can therefore suppose without loss of generality that Alice sends the extremal states $\omega_i$ with probabilities $p_i>0$ (all strictly positive), where $i=1,...,m$.  We denote the corresponding channel
 $C=\{\omega_i,p_i,e_y\}$. We consider the case $m>3$, and will show that there exists a channel with $m=3$ with capacity at least as big as that of $C$.

Denote by $\omega=\sum_{i=1}^m p_i \omega_i$ the average state sent by Alice.
By Caratheodory's theorem~\cite{Barvinok} $\omega$ can be written as a convex combination of at most three 
$\omega_i$: $\omega=\sum_{j\in J} q_j \omega_j$ where $J\subset \{1,...,m\}$, $\vert J \vert \leq 3$.
Let $q= \min_{j\in J} \frac{p_j}{q_j}$. We have $1>q>0$ (because all the $p_i>0$, all the $q_j>0$, and $\vert J \vert <m$) . 

We can rewrite 
$$\omega = q \left( \sum_{j\in J} q_j \omega_j\right)+ (1-q)  \left( \sum_{i=1}^m \frac{p_i-q q_i}{1-q} \omega_i \right)$$
where we set $q_i=0$ when $i\notin J$.
From the definition of $q$, it follows that both terms in parenthesis sum to $\omega$, and that the coefficients $(p_i-q q_i)/(1-q) \geq 0$ are positive, with at least one value of $i$ such that $p_i-q q_i = 0$.

By recurrence we can write 
\begin{equation} 
\omega=\sum_k q_k \left( \sum_{j\in J_k} q_{j\vert k} \omega_{j}\right)
\label{decomp1}
\end{equation}
where $J_k\subset \{1,...,m\}$, $\vert J_k \vert \leq 3$
and 
\begin{equation} 
\sum_{j\in J_k} q_{j\vert k} \omega_{j}=\omega ,
\label{decomp2}
\end{equation}
 and $q_k\geq 0$, $\sum_k q_k=1$, $q_{j\vert k}\geq 0$, $\sum_j q_{j\vert k} =1$.

We can now compare the capacities of the original channel 
$C=\{\omega_i,p_i,e_y\}$ and the channel $C'=\{\omega_j,q_{jk}=q_k q_{j\vert k},e_y\}$.
In the second channel, Alice first chooses $k$ with probability $q_k$, and then sends state $\omega_j$ with probability $q_{j\vert k}$, where we have that only three $q_{j\vert k}$ are non zero.

The overall probability distribution of channel $C'$ can be written as $p_{yjk}=p_{y\vert j k}  q_{j\vert k}  q_k$.
Eqs. (\ref{decomp1}, \ref{decomp2}) imply that $y$ is independent of $k$: $p_{y\vert k}=\sum_{j\in J_k} p_{y\vert j k}  q_{j\vert k} = \sum_{j\in J_k} q_{j\vert k}  \langle e_y, \omega_j \rangle =\langle e_y, \omega \rangle = p_y$.
The chain rule for mutual information then implies that the capacity for channel $C'$ is
$I(Y;J K)= I(Y;K)+I(Y;J\vert K)=I(Y;J\vert K)=\sum_k q_k I(Y;J\vert K=k)$.
where we have used that $Y$ is independent of $K$.
Hence there is at least one value of $k$ for which $I(Y;J\vert K=k)\geq I(Y;J K)$.
This value of $k$ corresponds to a channel with an alphabet of size 3.

On the other hand we can write $I(Y;J K)= I(Y;J)+I(Y;K\vert Y)\geq I(Y;J)$ where we can identify $I(Y;J)$ as the capacity of channel $C$. Hence there is as least one value of $k$ for which $I(Y;J\vert K=k) \geq I(Y;J)$, i.e. there is a channel with an alphabet of size 3 that has capacity greater or equal to the capacity of channel $C$.
\end{proof}

\end{document}